\documentclass[a4paper,12pt]{article}

\usepackage[top=1in, bottom=1in, left=1in, right=1in]{geometry}

\usepackage{algorithm}
\usepackage{algorithmic}

\usepackage{subfigure}
\usepackage{epsfig,amsmath,color,amsfonts,amsthm}
\usepackage{epsfig,color}

\usepackage{framed}
\usepackage{sidecap}
\usepackage{enumerate}

\newcommand{\pd}{\mathbf{P}}

\newcommand{\eps}{\varepsilon}

\newtheorem{theorem}{Theorem}[section]
\newtheorem{definition}[theorem]{Definition}

\newtheorem{lemma}[theorem]{Lemma}

\def\eps{{\epsilon}}

\newcommand{\ignore}[1]{}
\newcommand{\eat}[1]{}

\newcommand{\squishlist}{
 \begin{list}{$\bullet$}
  { \setlength{\itemsep}{0pt}
     \setlength{\parsep}{3pt}
     \setlength{\topsep}{3pt}
     \setlength{\partopsep}{0pt}
     \setlength{\leftmargin}{1.5em}
     \setlength{\labelwidth}{1em}
     \setlength{\labelsep}{0.5em} } }
\newcommand{\squishend}{
  \end{list}  }

\def\eps{{\epsilon}}

\begin{document}

\title{Distributed Computation of Mixing Time}


\author{Anisur Rahaman Molla\thanks{Research partially supported by ERC Grant No. 336495 (ACDC).} \\
	\small{Department of Computer Science}\\
	\small{University of Freiburg}\\
	\small{79110 Freiburg, Germany}\\
	\small{\texttt{anisurpm@gmail.com}}
\and 
Gopal Pandurangan \thanks{Supported, in part, by NSF grants  CCF-1527867,  CCF-1540512, and IIS-1633720.}\\
	\small{Department of Computer Science}\\
	\small{University of Houston}\\
	\small{Houston, Texas 77204, USA}\\
	 \small{\texttt{gopalpandurangan@gmail.com}}
}

\date{}
\maketitle

\begin{abstract}  The mixing time of a graph is an important metric, which is not only useful in analyzing connectivity and expansion properties of the network, but also serves as a key parameter in designing efficient algorithms. 
 We present an efficient distributed algorithm for computing the mixing time of  undirected graphs. Our algorithm  estimates the mixing time $\tau_s$ (with respect to a source node $s$) of any $n$-node undirected graph in $O(\tau_s \log n)$ rounds. 
Our algorithm is based on random walks and require very little memory and use lightweight local computations, and work in the {\em CONGEST} model. Hence our algorithm is scalable under bandwidth constraints and can be an helpful building block in the  design of topologically aware networks.    
\end{abstract}


\hspace{2.5in}


\section{Introduction}

Mixing time of a random walk in a graph is the time taken  by a  random walk to converge to the {\em stationary distribution} of the underlying graph.  It is an important parameter which is closely related to various key graph properties such as graph expansion, spectral gap, conductance etc. Mixing time is related to the {\em conductance} $\Phi$ and {\em spectral gap} ($1-\lambda_2$) of a $n$-node graph due to the known relations (\cite{JS89}) that $\frac{1}{1-\lambda_2}\leq \tau \leq \frac{\log n}{1-\lambda_2}$ and $\Theta(1-\lambda_2)\leq \Phi \leq \Theta(\sqrt{1-\lambda_2})$, where  $\lambda_2$ is the second largest eigenvalue of the adjacency matrix of the graph.   Small mixing time means the graph has  high expansion and spectral gap. Such a network supports fast random sampling (which
has many applications \cite{drw-jacm}) and low-congestion routing \cite{mihail}. Moreover, the spectral properties reveal a lot about the network structure \cite{DasSarmaMPU15}. Mixing time is also useful in designing efficient randomized algorithms in communication networks \cite{storage-spaa13,APR-podc13,SarmaMP15,DasSarmaMPU15,KM15,sirocco14}. 

There has been some previous work on distributed algorithms to compute mixing time.
The work of Kempe and McSherry \cite{kempe}  estimates the mixing time $\tau$ in $O(\tau \log^2 n)$ rounds. This approach uses {\em Orthogonal Iteration} i.e., heavy matrix-vector multiplication process, where each node needs to perform complex calculations and do memory-intensive  computations. This may not be suitable in
a lightweight environment.
 It is mentioned in their paper that it would be interesting whether a simpler and direct approach based on eigenvalues/eigenvectors can be used to compute mixing time.  Das Sarma et al. \cite{drw-jacm} presented a distributed algorithm based on sampling nodes by performing sub-linear time random walks and then comparing the distribution with stationary distribution. This algorithm can be sometimes faster than our approach, however, there is a grey area (in the comparison between the two distributions) for which their algorithm fails to estimate the mixing time with any good accuracy (captured
 by the accuracy parameter $\eps$ defined in
 Section \ref{sec:rwalk}). Our algorithm  is sometimes faster (when the mixing time is $o(\sqrt{n})$) and estimates the mixing time  with high accuracy (cf. Section \ref{sec:related}).

In this paper, we focus on developing a simple and efficient distributed algorithm for computing mixing time in graphs. 
Given an undirected $n$-node network $G$ and a source node $s$, our algorithm estimates the mixing time $\tau_s$ for the source node in $O(\tau_s \log n)$ rounds and achieves high accuracy of estimation. We note that this running time is non-trivial in the CONGEST model.\footnote{In the LOCAL model, all problems can be trivially solved in $O(D)$ rounds.}

 Our algorithm works in CONGEST model of distributed computation where only small-sized messages ($O(\log n)$-bits messages) are allowed in every communication round between adjacent nodes. Moreover, our algorithm is simple, lightweight 
 (low-cost  computations within a node) and easy to implement.  

Our approach crucially uses random walks. Random walks
are very  local and lightweight and require little index or state maintenance 
that makes it attractive to self-organizing networks \cite{BBSB04,ZS06}. Our approach, on a high-level, is based on efficiently performing many random walks from a particular node and computing the fraction of random walks that terminate over each node. We show that this fraction estimates the random walk probability distribution. Our approach achieves very high accuracy which is a requirement in some applications \cite{DasSarmaGP09,SarmaMP15,KM15,SpielmanT04}.

We remark that we can compute the mixing time of a graph in $O(m)$ time in the  CONGEST model, where $m$
is the number of edges in the graph. In the worst case, the mixing time $\tau$ could be $O(n^3)$ for some graph, e.g., the Lollipop graph. On the other hand, the complete graph topology can be collected to a single node (e.g., by electing a leader which takes $O(D)$ rounds \cite{KuttenPPRT15}) by flooding in $O(m)$ rounds and then the node  can compute the mixing time locally. The source node can compute the number of edges in the beginning (which can be done in $O(D)$ time) and then run the two algorithms in parallel and stop when one of them terminates. 


\subsection{Distributed Computing Model}
\label{sec:distmodel}
We model the communication network as an undirected, unweighted, connected graph $G = (V, E)$, where $|V| = n$ and $|E| = m$. Every  node has limited initial knowledge. Specifically, assume that each node is associated with a distinct identity number  (e.g., its IP address). 
At the beginning of the computation, each node $v$ accepts as input its own identity number and the identity numbers of its neighbors in $G$.
We also assume that the number of nodes and edges i.e., $n$ and $m$ (respectively) are given as inputs. (In any case, nodes can compute them easily through broadcast in $O(D)$ time; this does not affect  the asymptotic bounds of our algorithm.) The nodes are only allowed to communicate through the edges of the graph $G$. We assume that the communication occurs in  synchronous  {\em rounds}. 
We will use only small-sized messages. In particular, in each round, each node $v$ is allowed to send a message of size $O(\log n)$ bits through each edge $e = (v, u)$ that is adjacent to $v$.  The message  will arrive to $u$ at the end of the current round. 
This is a  widely used  standard model known as the {\em CONGEST model} to study distributed algorithms (e.g., see \cite{peleg,PK09}) and captures the bandwidth constraints inherent in real-world computer  networks. 

We  focus on minimizing the  {\em the running time}, i.e., the number of {\em rounds} of distributed communication. Note that the computation that is performed by the nodes locally is ``free'', i.e., it does not affect the number of rounds; however, we will only perform polynomial cost computation locally (in particular, very simple computations) at any node. 

\subsection{Random Walk Preliminaries}
\label{sec:rwalk}
We consider a {\em simple random walk} in an undirected graph: In each step the walk goes from the current node to a random neighbor i.e., from the current node $v$, the probability of moving to node a $u$ is $\Pr(v,u) = 1/d(v)$ if $(v,u) \in E$, otherwise $\Pr(v,u) = 0$, where $d(v)$ is the degree of $v$.

Suppose a random walk starts at some vertex $v$ in a graph $G$. Let $\pd_0$ be the initial distribution with probability $1$ at the node $v$ and zero at all other nodes. Then we get a probability distribution $\pd_t$ at time $t$ starting from the initial distribution $\pd_0$. Note that we hide the starting node in the notation of the probability distribution $\pd_t$. We hope that it is clear to the reader from the context. (Informally) we say that the distribution $\pd_r$ is stationary (or steady-state) for the graph $G$ when no further changes on the distribution i.e.,  $\pd_{r+t} = \pd_r$ for $t \ge 1$. It is known that the stationary distribution of an undirected connected graph is a well defined quantity which is $\bigl(\frac{d(v_1)}{2m}, \frac{d(v_2)}{2m}, \ldots, \frac{d(v_n)}{2m}\bigr)$, where $d(v_i)$ is the degree of the node $v_i$.  
We denote the stationary distribution vector by $\pi$, i.e., $\pi(v) = d(v)/2m$ for every node $v$. The stationary distribution of a graph is fixed irrespective of the starting node of a random walk, however, the time to reach to the stationary distribution could be different for the different starting nodes. 

The {\em mixing time} of a random walk starting from the source node $v$, denoted by $\tau_v$, is the  time (or number of steps) taken to reach to the stationary distribution of the graph.  The mixing time, denoted by $\tau$, is the maximum mixing time among all (starting) nodes in the graph.
The formal definitions are given below.

\begin{definition}\label{def:mixing-time} ($\tau_v(\eps)$--mixing time for the source $v$ and $\tau(\eps)$--mixing time of the graph)\\
Define $\tau_v (\eps)= \min \{t : ||\pd_t - \pi||_1 \leq \eps\}$, where $||\cdot||_1$ is the $L_1$ norm. Then $\tau_v(\eps)$ is called the $\eps$-near mixing time for any $\eps$ in $(0, 1)$. The mixing time of the graph is denoted by $\tau (\eps)$ and is defined by $\tau(\eps) = \max \{\tau_v(\eps): v \in V\}$. It is clear that $\tau_v (\eps) \leq \tau (\eps)$. 
\end{definition}

We note that  mixing time definition usually takes $\eps$ to be $1/2e$. Our goal is to estimate $\tau_v(\eps)$ for any small $\eps \in (0,1)$, say $\eps = 1/n^2$. We omit $\eps$ from the notation when it is understood from the context. The definition of $\tau_v$ is consistent due to the following standard monotonicity property of the random walk probability distributions.
\begin{lemma}\label{lem:monotonicity}
$||\pd_{t+1} - \pi||_1 \leq  ||\pd_t - \pi||_1$.
\end{lemma}
\begin{proof}(adapted from Exercise 4.3 in \cite{Levin})
The monotonicity follows from the fact that $||A\pd||_1 \le ||\pd||_1$, where $A$ is the transpose of the transition probability matrix of the graph and $\pd$ is any probability vector. That is, $A(i,j)$ denotes the probability of transitioning from the node $j$ to the node $i$. This in turn follows from the fact that the sum of entries of any column of $A$ is 1.

We know that $\pi$ is the stationary distribution of the transition matrix $A$. This implies that if $\ell$ is $\eps$-near mixing, then $||A^{\ell}\pd_0 - \pi||_1 \leq \eps$, by definition of $\eps$-near mixing time and $\pd_{\ell} = A^{\ell}\pd_0$. Now consider $||A^{\ell+1}\pd_0 - \pi||_1$. This is equal to $||A^{\ell+1}\pd_0 - A\pi||_1$, since $A\pi = \pi$.  However, this reduces to $||A(A^{\ell}\pd_0 - \pmb{\pi})||_1 \leq ||A^{\ell}\pd_0 - \pmb{\pi}||_1 \leq \eps$, (from the fact  $||A\pd||_1 \le ||\pd||_1$). Hence, it follows that $(\ell+1)$ is also $\eps$-near mixing time.
\end{proof} 
\noindent {\bf Mixing time computation problem.}
Given an undirected, connected and non-bipartite graph $G$, the goal is to design an efficient distributed algorithm to compute the mixing time of the graph. 

\subsection{Related Work}
\label{sec:related}
Das Sarma et al. \cite{drw-jacm} presented a fast decentralized algorithm for estimating mixing time, conductance and spectral gap of the network. In
particular, they show that given a starting node $s$, the mixing time with respect to $s$, i.e, $\tau_s$, can be
estimated in $\tilde{O}(n^{1/2} + n^{1/4}\sqrt{D\tau_s})$ rounds. This gives an alternative algorithm to the only previously known
approach by Kempe and McSherry \cite{kempe} that can be used to estimate
$\tau_s$ in $\tilde{O}(\tau_s)$ rounds. In fact, the work of \cite{kempe} does more and gives a decentralized algorithm for
computing the top $k$ eigenvectors of a weighted adjacency matrix
that runs in $O(\tau\log^2 n)$ rounds if two adjacent nodes are allowed to exchange $O(k^3)$ messages per round, where $\tau$ is
the mixing time and $n$ is the size of the network.  

The algorithm of Das Sarma et al. 
\cite{drw-jacm} is based on sampling nodes by performing sub-linear time random walks of certain length and comparing the distribution with the stationary distribution. Their algorithm can be faster than our approach in certain cases, but slower
in some cases than ours. In particular, if $\tau$ is smaller than $\max \{\sqrt{n}, n^{1/4} \sqrt{D}\}$, then
our algorithm is faster.  Also there is a grey area for the accuracy parameter $\eps$ for which their algorithm cannot estimate the mixing time. They use a testing result from Batu et al. \cite{BFFKRW} to determine if the two distributions are close enough. This test may fail if the difference between the two distributions falls in a certain interval. More precisely, the algorithm
of Das Sarma et al. estimates the mixing time for accuracy parameter $\eps = 1/(2e)$ with respect to a source node $v$, $\tau_v(1/2e)$  as follows: the estimated value will be between the true value  and $\tau_v(O(1/(\sqrt{n}\log n)))$. In contrast, our algorithm estimates
$\tau_v(\eps)$, for even small values of $\eps$, say $\eps = 1/n^2$: the estimated value will be
between the true value and $\tau_v(1/n^2)$. 

Random walks have been used in a variety of distributed network applications, we refer to \cite{storage-spaa13,APR-podc13,SarmaMP15,DasSarmaMPU15,drw-jacm,kempe,KM15,sirocco14} and the references therein for more details.

\section{Algorithm for Mixing Time}\label{sec:estimate-mixing}
We present an algorithm to compute the mixing time of a graph $G$ from a specified source node. In other words,  we present an algorithm which finds a length $\ell$ such that the probability distribution of a random walk of length $\ell$ reaches close to the stationary distribution. The main idea of our algorithm is to perform many random walks from the source node of some length $\ell$  in parallel. After $\ell$ steps, every node $u$ estimates the probability distribution $\pd_{\ell}(u)$ as the fraction of random walks that terminate at $u$ over all the walks. Then we compare the estimated distribution of $\pd_{\ell}$ with the stationary distribution $\pi$ to determine if they are sufficiently close; otherwise, we double the length $\ell$ and retry. Once we find the correct consecutive lower and upper bound of the length, a binary search will determine the mixing length (up to the allowed accuracy parameter $\epsilon$). The monotonicity property (cf. Lemma~\ref{lem:monotonicity}) admits the binary search.

Our algorithm starts with $\ell=1$ and runs $K= 80 n^8 \log n$ random walks of length $\ell$ from a source node, say node $s$. When the difference (i.e., the $L_1$-norm difference) between the estimated $\ell$-length walk distribution of $\pd_{\ell}$ with the stationary distribution $\pi$ is greater than $\eps$, $\ell$ is doubled and retried. This process is repeated to identify the largest $\ell$ such that the difference is greater than $\eps$ and the smallest $\ell$ such that the difference is less than $\eps$. These give lower and upper bounds on the required $\tau_s$ respectively. In fact, the upper bound (i.e., the smallest $\ell$ such that the difference is less than $\eps$) is at most twice as $\tau_s$, since we are doubling the length each time. Then a binary search between lower and upper bounds will determine the exact mixing time $\tau_s$. Note that the monotonicity property (cf. Lemma~\ref{lem:monotonicity}) guarantees that once the difference becomes less than $\eps$, then it would be always less for any larger length. We show in the analysis that the estimated probability of $\pd_{\ell}(u)$ is close to the actual probability for every node $u$ in each step of the walk (cf. Lemma \ref{lem:play-with-chernoff}).  The pseudocode is given in Algorithm~\ref{alg:mixing-time}.

We show (in the next section) that the above algorithm estimates mixing time accurately with high probability\footnote{With high probability means with probability at least $1 - \frac{1}{n}$.}. The main technical challenge in implementing the above method in CONGEST model is that performing many walks from a source node in parallel can create a lot of congestion. Our algorithm uses a crucial property of random walks to overcome the congestion (see e.g., \cite{DasSarmaMPU15,drw-jacm}). In particular, we show that there will be no congestion in the network even if we perform up to a polynomial (in $n$) number of random walks from the source node in parallel (cf. Lemma \ref{lem:congestion}). The basic idea is to send the {\em count} of the number of random walks that pass through an edge. As random walks are memoryless processes, it is sufficient to send the number of walks traversing an edge in a given round to estimate $\pd_{\ell}$. Since this number is polynomially bounded, $O(\log n)$ bits suffice. Therefore, it is easy to see that performing $\ell$-length of random walks finishes in $O(\ell)$ rounds in {\em CONGEST} model. We show that our algorithm computes mixing time accurately in $O(\tau_s\log n)$ rounds with high probability (cf. Theorem \ref{thm:mix-time}).   


\subsection{Reducing the number of random bits}
\label{sec:rrb}
To perform $K = 80 n^8 \log n$ random walks can require at least so many random bits per node (as well as that much time), if done in a straightforward manner. We can  make the algorithm more lightweight by doing only $O(d(u) \log n)$ work 
and only so many random coin tosses per node. 
Instead of doing coin-flips for each of the tokens separately, we do the following to reduce the number of random coin-flips overhead  at each node. Nodes which have less than $(degree\times \log n)$ tokens, will select an edge randomly for each token i.e., they will do the coin-flips for each token.  At any round, if a node $u$ has tokens $T^u$ such that $T^u\geq d(u)\log n$, then $u$ forwards the average $T^u/d(u)$ tokens to all of its neighbors (there are $d(u)$ many). We show that this simple averaging approach has the bias of at most $O(\log n)$ with high probability. In fact, if we perform coin-flips for each token i.e., select a neighbor randomly for each token then, in expectation, $T^u/d(u)$ tokens will choose to go to one neighbor, where $T^u$ is the total number of tokens at $u$. Then by a Chernoff bound, it is easy to show that the deviation from the mean  is at most  $O(\log n)$  with high probability. That is, the bias is $O(\log n)$ tokens per edge with high probability.  Therefore, the total bias at a node $u$ is $O(d(u)\log n)$ in a single round ($O(\log n)$ bias coming from $d(u)$ neighbors). Hence, total bias at node $u$ until the algorithm stops is $O(\tau_s d(u) \log n)$, since the algorithm runs for $O(\tau_s)$ rounds. This quantity is at most $n^4 \log n$, since $d(u) < n$ and $\tau_s \leq O(n^3)$ in worst case. Therefore,  the deterministic averaging method has the additional approximation error at most $(n^4 \log n)/K = 1/n^4$, which is negligible compared to the considered error $\eps = 1/n^2$. 

\begin{algorithm}[H]
\caption{\sc EstimateMixingTime}
\label{alg:mixing-time}
\textbf{Input:} A graph $G = (V, E)$ and a source node $s$. \\
\textbf{Output:} $\tau_s (1/n^2)$ (mixing time starting from the node $s$).

\begin{algorithmic}[1]
\STATE Node $s$ creates a BFS tree via flooding (rooted at $s$) so that each node knows their parent in the BFS tree.  
\STATE Node $s$ broadcasts the value $K = 80 n^8 \log n$ to all the other nodes via broadcasting over the tree.  
\FOR{$h =0, 1, 2, \dots$}
\STATE \label{stp:creat-token} $\ell \leftarrow 2^{h}$   
\STATE \label{stp:rw-prob-computation} Node $s$ creates $K = 80 n^8 \log n$ random walk tokens.
\FOR{round $i = 1, 2, \ldots, \ell$}
\STATE Each node $v$ holding at least one token, does the following in parallel: 
\STATE For every neighbor $w$, set $T^v_w = 0$ \hspace{0.1in} // [$T^v_w$ indicates the number of tokens chosen to move to $w$ from $v$ at round $i$] 
\STATE For all the tokens calculate the number $T^v_w$ for each neighbor $w$ (as described in Section \ref{sec:rrb}).
\STATE Send the counter number $T^v_w$ to the neighbor $w$.  
\ENDFOR  
\STATE \label{stp:end-rw-prob} Each node $w$ counts the total number of tokens it holds. Say, $\zeta_w = \sum_{v \in N(w)} T^v_w$. 
\STATE Each node $w$ sends the absolute difference value $\partial_w = |\frac{\zeta_w}{K} - \frac{d(w)}{2m}|$ to node $s$ through the BFS tree. Each node sums up the $\partial_w$ values of their children and forwards the sum to its parent. \label{stp:bfs-upcast}
\STATE Node $s$ locally checks:
\IF{ $(\sum_{w \in V} \partial_w \leq 1/n^2)$} \label{stp:checking}
\STATE BREAK;
\ELSE
\STATE Continue from Step~\ref{stp:creat-token} for the next (doubling) value of $\ell$.
\ENDIF
\ENDFOR
\STATE Let's say $\ell = 2^l$ for some integer $l$. 
\STATE Perform binary search from $2^{l-1}$ to $2^l$ to find an integer  $\ell$ such that:
\STATE  For length $\ell$, $\sum_{w \in V} \partial_w \leq 1/n^2$, and for  length $\ell -1$, $\sum_{w \in V} \partial_w > 1/n^2$,  by computing from Step~\ref{stp:rw-prob-computation} to Step~\ref{stp:checking}.
\STATE Output $\ell$.
\end{algorithmic}
\end{algorithm}

\subsection{Analysis of the Algorithm}\label{sec:analysis}

We first show a result which will imply that our algorithm correctly estimates the mixing time for a given source node. Let $\pd_{\ell}(u)$ be the probability that a random walk of length $\ell$ terminates over node $u$, starting from the given source node. We show that the above process of performing many random walks and then computing the fraction of walks that stop at $u$ can approximate $\pd_{\ell}(u)$ with high accuracy for every node $u$. For this, we can ignore the probability $\pd_{\ell}(u)$ whose value is less than $1/n^4$. Because in the worst case, there may be at most $n$ such vertices and the sum of these ignored probabilities is less than $1/n^3$, which is negligible compared to $1/n^2$, the estimation error. (Here we assume that the mixing time estimation error $\eps$ is  $1/n^2$ --- see Definition~\ref{def:mixing-time}). However, one can achieve much more accuracy by performing a larger number of random walks. In the following lemma, we show that if we perform $K = 80 n^8 \log n$ random walks from a source node $s$, then the above algorithm can estimate $\pd_{\ell}(u)$ at any node $u$ for which $\pd_{\ell}(u) \geq 1/n^4$.  

\begin{lemma}\label{lem:play-with-chernoff}
If the probability of an event $X$ occurring is $p$ such that $p \geq 1/n^4$, then in $t = 80 n^8 \log n$ trials , the fraction of times the event $X$ occurs is $p \pm \frac{1}{n^6}$ with high probability.  
\end{lemma}    
\begin{proof}
The proof follows from the standard Chernoff bound, $$ \Pr \left[\frac{1}{t} \sum_{i=1}^t X_i < (1 - \delta)p \right] < \left(\frac{e^{-\delta}}{(1-\delta)^{(1-\delta)}} \right)^{tp} < e^{-tp\delta^2/2}$$ and 
$$\Pr \left[\frac{1}{t} \sum_{i=1}^t X_i > (1 + \delta)p \right] < \left(\frac{e^{\delta}}{(1+ \delta)^{(1+ \delta)}} \right)^{tp},$$ where $X_1, X_2, \ldots, X_t$ are $t$ independent identically distributed $0-1$ random variables such that $\Pr[X_i = 1] = p$ and $\Pr[X_i = 0] = (1-p)$. The right hand side of the upper tail bound further reduces to $2^{-\delta t p}$ for $\delta > 2e -1$ and for $\delta <2e - 1$, it reduces to $e^{-tp\delta^2/4}$. 

Let $\delta =  \frac{1}{n^2}$. Then $\delta < 1$. So we consider the weaker bound of both the lower and upper tail bounds which is $e^{-tp\delta^2/4}$. Therefore, by choosing $t = 80 n^8\log n$, we get  $e^{-tp\delta^2/4} \leq e^{- \frac{t}{4n^8}} = e^{- 20\log n} = 1/n^{20}$. This implies that $\frac{1}{t} \sum_{i=1}^t X_i < p - 1/n^6$ and $\frac{1}{t} \sum_{i=1}^t X_i > p + 1/n^6$ with probability at most $1/n^{20}$, since $\delta p \geq 1/n^6$. Therefore, $X = \frac{1}{t} \sum_{i=1}^t X_i \in [p - 1/n^6,\, p + 1/n^6]$ with high probability.
\end{proof}


\begin{lemma}\label{lem:correctness}
The algorithm {\sc EstimateMixingTime} approximates the probability $\pd_{\ell}(u)$ for every length $\ell$ and for all node $u$ (such that $\pd_{\ell}(u) \geq 1/n^4$) with high probability. 
\end{lemma}
\begin{proof}
Let's first consider a particular length $\ell$. Suppose the algorithm {\sc EstimateMixingTime} outputs the estimated probability $\tilde{\pd}_{\ell}(u)$ for each node $u$ for the length $\ell$. It is follows from the Lemma \ref{lem:play-with-chernoff} that $\tilde{\pd}_{\ell}(u) = \pd_{\ell}(u) \pm 1/n^6$ with high probability for any node $u$ for which $\pd_{\ell}(u) \geq 1/n^4$ (by taking union bound). Therefore,  $\mid \tilde{\pd}_{\ell}(u) - \pd_{\ell}(u) \mid \leq 1/n^6$ for any node $u$ (such that $\pd_{\ell}(u) \geq 1/n^4$). Since $\ell$ can be arbutrary, the lemma holds for every length $\ell$.   
\end{proof}

Before proceeding to the main result (cf. Theorem \ref{thm:mix-time}) of this section, we prove a crucial lemma on the congestion of our algorithm. The lemma below guarantees that there will be no congestion even if we perform a polynomial number of random walks in parallel in the network.   

\begin{lemma}\label{lem:congestion}
If every node perform at most a polynomial (in $n$) number of random walks in parallel then there will be no congestion in the network. 
\end{lemma}
\begin{proof}
It follows from our algorithm that each node only needs to know the number of random walks that stop over itself after $\ell$ rounds. Therefore nodes do not require to know from which source node or rather from where it receives the random walk tokens. Hence it is not needed to send the ID of the source node with the tokens. Recall that in our algorithm,  in each round, every node currently holding at least one random walk token (could be many) does the following.  For each token, an edge (i.e., a neighbor) is chosen uniformly at random to send the token. A particular edge may be chosen for multiple tokens. Instead of sending each token separately through that edge, the algorithm simply sends the count, i.e., number of tokens chosen to pass over the edge. Random walk is a Markovian process which means the process is memoryless. Therefore, it is sufficient to send the count of walks traversing an edge, without the need to append any additional information about the walks. Since we consider {\em CONGEST} model, a polynomial in $n$ number of token's count can be sent in a single piece of message of size $O(\log n)$ without any congestion over edges.  
\end{proof}

\begin{theorem}\label{thm:mix-time}
Given a graph $G$ and a source node $s$, the algorithm {\sc EstimateMixingTime} takes $O(\tau_s \log n)$ rounds and finds the mixing time $\tau_s(1/n^2)$ with high probability. 
\end{theorem}
\begin{proof}
\noindent {\bf Correctness.}
The Lemma \ref{lem:correctness} says that when $\ell$ reaches the mixing time of the graph, the estimated probability $\tilde{\pd}_{\ell}(u)$ of $\pd_{\ell}(u)$ will be close to the stationary distribution $d(u)/2m$ for each node $u$. Hence, the $L_1$-norm difference is $\sum_{u\in V}\mid \tilde{\pd}_{\ell}(u) - d(u)/2m \mid \leq 1/n^5 + 1/n^3 < 1/n^2$. The term $1/n^3$ is coming from the fact that we neglected those probabilities such that $\pd_{\ell}(u) < 1/n^4$; so summing up all of them could be at most $1/n^3$. Further, it follows from the monotonicity property (cf. Lemma~\ref{lem:monotonicity}) that once the norm difference becomes less than $1/n^2$, it would always be less than $1/n^2$ for any larger length of the random walk. Hence, our algorithm terminates when
 the mixing time is correctly estimated (cf. Definition~\ref{def:mixing-time}). 

\noindent {\bf Running time.}
To estimate the mixing time, we try out
increasing values of $\ell$ that are powers of $2$.  Once we find the
right consecutive powers of $2$, the monotonicity property admits a
binary search to determine the exact value. Since each node $u$ knows its own stationary probability (determined just by its
degree), they can compute the difference $\partial_u = |\frac{\zeta_u}{K} - \frac{d(u)}{2m}|$ locally and send it to the source node $s$. Recall that $\zeta_u$ is the number of walks that stop at $u$ after $\ell$ steps and $K$ is the total number of random walks started initially. The node $s$ eventually collects the sum of all $\partial_u$s and checks if $\sum_{u \in V} \partial_u$ is less than $1/n^2$. Performing $K$ random walks of length $\ell$ can be done in $O(\ell)$ rounds, as there is no congestion (cf. Lemma \ref{lem:congestion}). Further, sending the sum of all $\partial_u$s to the source node $s$, can be done in $O(D)$ rounds ($D$ is the diameter). Since this can be done by an upcast through the BFS tree rooted at $s$ (see Step~\ref{stp:bfs-upcast} of Algorithm~\ref{alg:mixing-time}). Hence, for a particular $\ell$, it requires $O(\ell + D)$ rounds. Therefore, total time required to compute the mixing time for a source node $s$ is $O(\sum_{i = 1}^{\log \tau_s} (\ell + D))$ rounds which is $O((\tau_s + D)\log n)$ rounds as $\ell$ is at most $\tau_s$ which is polynomially bounded. Notice that when $\ell \geq \tau_s(1/n^2)$, then the $L_1$ difference between estimated distribution and stationary distribution becomes less than $1/n^2$ (correctness follows from the Lemma \ref{lem:correctness}). This gives the estimated mixing length at most twice as $\tau_s (1/n^2)$, since we are doubling the length each time. Then a binary search will determine the exact mixing time $\tau_s (1/n^2)$. The binary search takes at most $O((\tau_s + D)\log n)$ rounds. Therefore, the algorithm finishes in $O(\tau_s \log n)$ rounds, since $D \leq \Omega(\tau_s)$ for any graph.
\end{proof}

\section{Conclusion} 
We presented a random-walk based distributed algorithm with provable guarantees to compute the mixing time  of undirected graphs. Our algorithm is simple and lightweight, and estimates the mixing time with high accuracy for all ranges of mixing time.
Since mixing time is an important parameter with close relation to spectral properties of the network, our algorithm can be a basic building block in developing topologically-aware networks that measure and monitor their network
properties.

\bibliographystyle{abbrv}
\bibliography{Distributed-RW-arxiv.bib}

\end{document}